\documentclass[11pt]{article}
\usepackage{times}
\usepackage{setspace}
\usepackage{amsfonts}
\usepackage{amsmath}
\usepackage{amsthm}
\usepackage{amssymb}
\usepackage{oldgerm}
\usepackage[]{graphicx}

\newcounter{aid}

\theoremstyle{plain}
\newtheorem{theorem}{Theorem}
\newtheorem{corollary}{Corollary}

\newtheorem{lemma}{Lemma}
\newtheorem{definition}{Definition}

\theoremstyle{definition}
\newtheorem{remark}{Remark}

\def\qed{\hfill $\square$}

\DeclareMathOperator{\supp}{supp}
\DeclareMathOperator{\card}{card}

\def\R{\mathbb{R}}
\def\N{\mathbb{N}}

\begin{document}

\bibliographystyle{ieeetralpha}

\title{Entropy Density and Mismatch in  High-Rate  Scalar\\
  Quantization with R\'enyi Entropy Constraint} 

{\renewcommand{\thefootnote}{}
\footnotetext{\hspace{-0.3cm} W.\ Kreitmeier is with the 
Department of Informatics and Mathematics, University of Passau, 
Innstra\ss e 33, 94032 Passau, Germany (email:
\texttt{wolfgang.kreitmeier@uni-passau.de}). 
T.  Linder is with the Department of
Mathematics and Statistics, Queen's University, Kingston, Ontario,
Canada K7L 3N6 (email: {\tt linder@mast.queensu.ca}). }

\footnotetext{\hspace{-0.3cm} This research was supported in part by
  the German Research Foundation (DFG) and the Natural Sciences and
  Engineering Research Council (NSERC) of Canada.} }

\author{Wolfgang Kreitmeier and Tam\'as Linder} 
\date{\today}

\maketitle

\begin{abstract}  
Properties of scalar quantization with $r$th power distortion and
constrained R\'enyi entropy of order $\alpha\in (0,1)$ are
investigated.  For an asymptotically (high-rate) optimal sequence of
quantizers, the contribution to the R\'enyi entropy due to source
values in a fixed interval is identified in terms of the ``entropy
density'' of the quantizer sequence.  This extends results related to
the well-known point density concept in optimal fixed-rate
quantization. A dual of the entropy density result quantifies the
distortion contribution of a given interval to the overall
distortion. The distortion loss resulting from a mismatch of source
densities in the design of an asymptotically optimal sequence of
quantizers is also determined. This extends Bucklew's fixed-rate
($\alpha=0$) and Gray \emph{et al.}'s variable-rate ($\alpha=1$)
mismatch results to general values of the entropy order parameter
$\alpha$.

\end{abstract}


\allowdisplaybreaks

\section{Introduction}

Asymptotic quantization theory studies the performance of quantizers
of a fixed dimension in the limit of high rates (low distortion). This
approach complements Shannon's rate-distortion theory where optimal
codes of a fixed rate (distortion) are investigated as the dimension
becomes asymptotically large.  Panter and
Dite \cite{PaDi51} were the first to derive a formula for the mean
square distortion of  optimum scalar quantizers as the number of
quantization levels becomes asymptotically large.  Zador's classic
work \cite{Zad63} for vector quantizers 
determined the asymptotic behavior of the minimum quantizer distortion
under a constraint on either the log-cardinality of the quantizer
codebook (fixed-rate quantization) or the Shannon entropy of the
quantizer output (entropy-constrained quantization).  Zador's results
were later clarified and generalized by Bucklew and Wise \cite{BuWi82}
and Graf and Luschgy \cite{GrLu00} for the fixed-rate case, and by
Gray \emph{et al}.\ \cite{GrLiLi02} for the entropy-constrained case.
Gray and Neuhoff \cite{GrNe98} provide a historical overview of
related results.

One way to unify and extend the fixed and variable-rate results is to
define the quantizer's rate by the R\'enyi entropy
of order $\alpha$ of its output. This generalized rate concept  includes the
traditional rate definitions as special cases since $\alpha=0$
corresponds to fixed-rate quantization, while setting $\alpha=1$
yields variable-rate quantization. This approach was first suggested
in \cite{GrLiGi08} as an alternative to the Lagrangian rate definition
considered there which simultaneously controls codebook size and output
(Shannon) entropy. Further motivation for using R\'enyi entropy as
quantization rate can be obtained from axiomatic considerations
\cite{Ren60b,AcDa75}, as well as from the operational role of the
R\'enyi entropy in variable-length lossless coding
\cite{Cam65,Jel68,Bae03}.

The theory of quantization with R\'enyi $\alpha$-entropy constraint
has recently been explored in \cite{Kre10a,Kre10b,Kre11,KrLi11}.  In
particular, \cite{Kre10b} derived the sharp asymptotic behavior of the
$r$th power distortion of optimal $d$-dimensional vector quantizers
for $\alpha \ge 1+r/d$.  In \cite{KrLi11} the technically more
challenging $\alpha < 1$ case was considered and the asymptotically
optimal $r$th power distortion was determined for scalar quantization
($d=1$) and a fairly large class of source densities. Thus, at least
for scalar quantization, only the case $\alpha \in (1,1+r)$ remains
open, and it is conjectured in \cite{KrLi11} that the main result
there remains valid in this range of the parameter $\alpha$.

In addition to the asymptotic behavior of the optimal quantizer
performance, asymptotic quantization theory has also been concerned
with more subtle properties of (asymptotically) optimal
quantizers. One such property is the existence, for a sequence of
quantizers, of the so-called quantizer point density function, loosely
defined as a probability density which, when integrated over a region,
gives the fraction of the quantization levels contained in that
region. More formally, a point density, if exists, is the
probability density function of the limit distribution of the output
levels of a sequence of quantizers.  Point densities and the closely
related companding quantizers have been instrumental in the early
pioneering investigations into optimal scalar and vector quantization
\cite{Ben48,PaDi51,Llo57,Ger79} (see also \cite{NaNe95} for a rigorous
reformulation of Bennett's result for the vector case and
\cite{GrNe98} for the history of these results). Bucklew \cite{Buc84}
was the first to rigorously establish the existence of the point
density function for an asymptotically optimal sequence of fixed-rate
quantizers.  To our knowledge, no such rigorous result is known for
variable-rate quantization.  The concept of quantizer point density
has been very useful in analyzing the performance of quantizers in a
distributed setting (e.g. \cite{TiHe04,MiGoVa11}).

Asymptotic quantization theory has also been successful in providing
mismatch results that quantify the loss in performance when a sequence
of quantizers that is asymptotically optimal for one source is applied to a
different source. Mismatch results are theoretically important and  in
practice they may provide a means for quantifying the performance of
code designs that are based on source models estimated from data.
For fixed-rate vector quantization Bucklew~\cite{Buc84} was the first to
prove such a rigorous mismatch result. The variable-rate analog of
this result was proved in \cite{GrLi03} where connections with
mismatch results in rate-distortion theory and robust lossy coding
were also pointed out. More recently, Na \cite{Na11} determined sharp
asymptotic formulas for variance-mismatched scalar quantization of
Laplacian sources. 

In this paper we extend some of the more refined results of fixed and
variable-rate asymptotic quantization theory in the framework of
quantization with R\'enyi entropy constraint of order $\alpha\in
(0,1)$. The concept of a quantizer point density is a problematic one
for (R\'enyi) entropy-constrained quantization since (near) optimal
quantizers can have an arbitrarily large number of levels in any
bounded region. Instead, we investigate the R\'enyi entropy
contribution of a given interval to the overall rate. One of our main
results, Theorem~\ref{theoremdens}, shows that for a large class of
source densities and an asymptotically optimal sequence of quantizers,
this contribution can be quantified by the so called entropy density
of the sequence. A dual of this result, Corollary~\ref{corollary123},
quantifies the distortion contribution of a given region to the
overall distortion in terms of the so-called distortion
density. Interestingly, it turns out that the entropy and distortion
densities are equal in the cases we investigate (Remark~\ref{rem_coinc}).  Our other main
contribution, Theorem~\ref{mresul}, is a mismatch formula for a
sequence of asymptotically optimal R\'enyi entropy constrained scalar
quantizers. From our density and mismatch results we can recover the
known results for the traditional rate definitions by formally setting
$\alpha=0$ or $\alpha=1$.

The rest of the paper is organized as follows. In the next section we
formulate the quantization problem and give a somewhat  informal
overview of our results 
in the context of prior work. In Section~\ref{sec_density} the entropy
and distortion density results are presented and proved. The
mismatch problem is 
considered in Section~\ref{sec_mismatch}. Concluding remarks are given
in Section~\ref{sec_concl}.

\section{Preliminaries and overview of results}

\label{sec_prelim}

\subsection{R\'enyi entropy and quantization}

We begin with the definition of R\'enyi entropy of order
$\alpha$. 
Let $\mathbb{N} := \{1,2,\ldots \}$  and  let 
$p=(p_{1},p_{2},\ldots) \in [0,1]^{\mathbb{N}}$ be a probability vector,
i.e. $\sum_{i=1}^{\infty} p_{i} = 1$. 
For any $\alpha\in (0,\infty)\setminus \{1\}$, the R\'enyi entropy
  of order $\alpha$,  
$\hat{H}^{\alpha}(p) \in [0,\infty ]$,
is defined as (see \cite{Ren60b} or  \cite{AcDa75})  
\[
\hat{H}^{\alpha}(p)=
\frac{1}{1-\alpha} \log \left(  \sum\limits_{i=1}^{\infty}
p_{i}^{\alpha} \right).
\]

\begin{remark}
\label{remarkhospital}
All logarithms in this paper have base $e$. Setting $0^0:=0$, we can extend the definition to $\alpha=0$, obtaining
\begin{equation}
\label{eq_h0def}
\hat{H}^{0}(p) = \log \left(  \card \{ i \in \mathbb{N}: p_{i} > 0 \} \right)
\end{equation}
where $\card$ denotes cardinality. Also, 
using the convention $0\log 0:=0$, it is easy to see that letting $\alpha\to 1$
yields  the regular (Shannon) entropy of
$p$:
\[
\hat{H}^{1}(p):=\lim_{\alpha \rightarrow 1} \hat{H}^{\alpha}(p )  =-\sum_{i=1}^{\infty} p_i\log p_i
\]
assuming $\hat{H}^{\alpha}(p)$ is finite for some  $\alpha<1$.
\end{remark}

Let $X$ be a real-valued random variable with distribution $\mu$.  Let
$\mathbb{I} \subset \mathbb{N}$ be an index set (thus $\mathbb{I}$ is
either  finite or countably infinite) and $\mathcal{S} = \{ S_{i} : i \in \mathbb{I} \}$ a  Borel
measurable partition of the real line $\mathbb{R}$. Moreover let
$\mathcal{C} = \{ c_{i} : i \in \mathbb{I} \}$ be set of distinct
points in $\mathbb{R}$. Then $(S_i,c_i)_{i\in \mathbb{I}}$ defines a
(scalar) \emph{quantizer} $q : \mathbb{R} \rightarrow \mathcal{C} $
such that
\[
q(x) = c_{i} \qquad \text{ if and only if } \qquad x \in S_{i} .
\]  
We call $\mathcal{C}$ the \emph{codebook} and the $c_i$ the
codepoints (or quantization levels). Each $S_{i} \in \mathcal{S}$ is called
\emph{codecell}. Clearly  $\mathcal{C} = q (\mathbb{R})$  is the
range of $q$ and 
\[
\mathcal{S} = \{ q^{-1}(z) : z \in q(\mathbb{R})  \}
\]
where $q^{-1}(z)=\{x\in \R: q(x)=z\}$.
Let $\mathcal{Q}$ denote the set of scalar quantizers,
i.e., the set of all Borel-measurable mappings $q : \mathbb{R}
\rightarrow \mathbb{R}$ with a countable range.  The discrete random variable $q(X)$ is a quantized
version of the random variable $X$. With any enumeration $\{ i_{1}, i_{2}, \ldots \}$ of $\mathbb{I}$ we define
\[
H^{\alpha}_{\mu }(q) = 
\hat{H}^{\alpha } ( \mu ( S_{i_{1}} ), \mu ( S_{i_{2}} ),\ldots ) 
\]
as the R\'enyi entropy of order $\alpha$ of $q$ with respect to $\mu$.
Thus  $H^{0}_{\mu }(q)$ is the log-cardinality of the codebook
of $q$ (we assume without loss of generality that each codecell of $q$
has positive probability)  and $H^{1}_{\mu }(q)$ is the Shannon entropy of the quantizer
output.

For $r\ge 1$  and  $q \in \mathcal{Q}$ 
we measure the approximation  error between $X$ and $q(X)$  by the
   the $r$th power distortion defined by 
\[
D_{\mu}(q)= E | X-q(X)|^r  =
\int | x-q(x)|^r \,  d \mu (x).
\]
For any $R \geq 0$ we define
\[
D_{\mu}^{\alpha }(R)= \inf \{ D_{\mu}(q) : q \in \mathcal{Q},
H^{\alpha }_{\mu }(q) \leq R \} 
\]
the optimal quantization distortion of $\mu$ under R\'enyi entropy
constraint $R$. We call a quantizer $q$ optimal for $\mu$ under the entropy constraint
$R$ if $D_{\mu}(q)=D_{\mu}^{\alpha }(R)$ and $H^{\alpha }_{\mu }(q)
\le R$.  In particular, $D_{\mu}^0(R)$
is the minimum distortion of any quantizer with codebook size not exceeding $e^R$, while $D_{\mu}^1(R)$ is the minimum distortion under
Shannon entropy constraint $R$.

In  the rest of this paper  all distributions to be quantized will be
absolutely continuous with respect to the Lebesgue measure $\lambda$
on the real line. If such a distribution   $\mu$
has probability density function  $g$, then we will use the
notation $\mu=g\lambda$. 
We denote by $\supp (\mu )$ the support of $\mu$ (the smallest closed
set whose complement has zero $\mu$ measure). If $\mu=g\lambda$, then
we define $\supp(g)=\supp(\mu)$.  We will
also assume throughout the paper that the $r$th moment of $\mu$
is finite, i.e. $\int |x|^{r} \, d \mu (x) < \infty$. This condition
is sufficient (but not necessary) for $D_{\mu}^{\alpha}(R)$ to be
finite for all $R\ge 0$. 

It has been shown in \cite{Kre11} that under the above conditions, the
set of all quantizers $\mathcal{Q}$ in the definition of
$D_{\mu}^{\alpha}$ can be replaced by the set of quantizers having
finitely many codecells, each of which is an interval.  In view of
this, we will assume throughout the whole paper that the codecells of
every quantizer $q \in \mathcal{Q}$  are intervals (but we do not
restrict the number of codecells to be finite) and each codepoint  is
contained 
in the interior of the associated codecell. 

\subsection{Asymptotic optimality and conditional distributions}

\label{seq_asymptcond}

The main result of \cite{KrLi11} implies that under suitable
assumptions on the source density $g$, for all $\alpha\le 1$, 
\begin{equation}
 \label{eq_mainasympt}
   \lim_{R\to \infty} e^{rR} D_{\mu}^{\alpha}(R) 
=  \frac{1}{(1+r)2^r}  e^{r  h^{\beta_1}(g)}
\end{equation}
where  $ \beta_{1} = \frac{1- \alpha + \alpha r}{1- \alpha +
  r}$  and
\[
h^{\beta_1}(g)=
\frac{1}{1 - \beta_1} \log \biggl( \int  g^{\beta_1 } \,
d\lambda \biggr)
\]
is the R\'enyi differential entropy
of order $\beta_1$ of $g$.

We can formally recover Zador's classical
results \cite{Zad63} in the scalar setting from (\ref{eq_mainasympt}).
Letting $\alpha=0$, we have $\beta_1= \frac{1}{1+r}$ and $ e^{r
  h^{\beta_1}(g)} =(\int g^{\frac{1}{1+r}} \,
d\lambda)^{1+r} = \|g\|_{\frac{1}{1+r}}$, yielding Zador's formula
for fixed-rate scalar quantization. For $\alpha=1$, we have $\beta_1=
1$ and $e^{r h^{\beta_1}(g)} = e^{r h(g)}$, where $h(g)=-\int g\log
g\, d\lambda$ is the
Shannon differential entropy of $g$, and (\ref{eq_mainasympt})
becomes Zador's formula for variable-rate scalar quantization.
In view of (\ref{eq_mainasympt}) we call a sequence of quantizers
$(q_n)_{n\in \N}$  \emph{asymptotically optimal} if
$H^{\alpha }_{\mu }(q_n) \to \infty$ and
\[
\lim_{n\to \infty} e^{rH^{\alpha }_{\mu }(q_n)}  D_{\mu}(q_n) =
\frac{1}{(1+r)2^r}  e^{r  h^{\beta_1}(g)}.
\]

Suppose  $I$ is a bounded interval with  positive $\mu$
probability. We denote by
$\mu(\cdot|I)$  the conditional distribution for $\mu$ given
$I$ and  by $g_I$ the corresponding conditional density (so that
$\mu(\cdot|I) =g_I\lambda$).  We show in Theorem~\ref{theoremdens}
that for $\alpha\in (0,1)$ 
any quantizer sequence $(q_n)$ that is asymptotically optimal for
$\mu$ is also asymptotically optimal for $\mu(\cdot|I)$, i.e., 
\begin{equation}
\label{eq_condasympopt}
\lim_{n\to \infty} e^{rH^{\alpha }_{\mu(\cdot|I) }(q_n)}
D_{\mu(\cdot|I)}(q_n) = \frac{1}{(1+r)2^r} e^{r h^{\beta_1}(g_I)}.
\end{equation}
Although this result is not very surprising, it will be very useful in
establishing further, more subtle properties of asymptotically optimal
quantizers.

\subsection{Entropy and distortion densities}

\label{sec_entrdist}

Let 
$N_n(I)$ denote the  number of codepoints of
$q_n$  contained in  an interval $I$. Let $\alpha=0$ and let $(q_n)$ be a sequence
of asymptotically 
optimal $n$-level quantizers (so that $H^0_{\mu}(q_n)=\log
n$). Specialized to the scalar case, one
 important result of Bucklew \cite{Buc84} shows that 
\begin{equation}
\label{eq_pointdensbucklew}
\lim_{n\to \infty} \frac{N_n(I)}{n} = \frac{\int_I g^{\frac{1}{1+r}}
  \, d\lambda }{\int_{\R}
  g^{\frac{1}{1+r}} \, d\lambda}.
\end{equation}
Thus the probability density    $g^{\frac{1}{1+r}}/ \int
g^{\frac{1}{1+r}} \, d\lambda $ can be interpreted as the point
density function for the codepoints of asymptotically optimal
quantizers (see also \cite[Thm.~7.5]{GrLu00}). Point densities are
useful in gaining insight into the structure of (asymptotically)
optimal quantizers and can be used to construct such quantizers
via a companding construction. 

Unfortunately, no rigorous point density results are known for
$\alpha=1$.  In fact, even the definition of a point density function
is problematic for entropy-constrained quantization since for sources
with a density, at any rate $R>0$ there exist near-optimal quantizers
that have an arbitrarily large  number of  codepoints contained in a given
bounded interval. Thus an analog of (\ref{eq_pointdensbucklew}) cannot hold
for an arbitrary  sequence of asymptotically optimal quantizers, although
heuristic arguments indicate that under some structural restrictions
asymptotically optimal variable-rate quantizers have a uniform point
density (see, e.g., \cite{GiPi68,Ger79}).

To define a tractable analog of the point density function, recall
that 
$\mu(\cdot|I)$ denotes the conditional distribution for $\mu$ given
$I$. In view of (\ref{eq_h0def}), we have $N_n(I)=
e^{H^0_{\mu(\cdot|I)}(q_n)}$ and $n=e^{H^0_{\mu}(q_n)}$. Thus the fraction
of codepoints contained in $I$ on the left hand side of
(\ref{eq_pointdensbucklew}) can be rewritten as
\begin{equation}
\label{eq_contr}
 \frac{N_n(I)}{n} =  \frac{e^{H^0_{\mu(\cdot|I)}(q_n)}}{e^{H^0_{\mu}(q_n)}}.
\end{equation}
This ratio represents the relative contribution of the interval $I$ to
the total R\'enyi entropy of order $\alpha=0$. 

The interpretation in \eqref{eq_contr} motivates us to define the
R\'enyi entropy
contribution of an interval $I$ in a similar way for general $\alpha$. 
In Theorem~\ref{theoremdens}, we identify the  limit 
of this  entropy contribution: Under appropriate 
conditions on the source density, for any $\alpha\in (0,1)$ 
and  asymptotically optimal sequence $(q_n)$,  we have
\begin{equation}
\label{propentrloc0}
\lim_{n\to \infty} \frac{e^{(1 - \alpha ) H_{\mu ( \cdot | I
      )}^{\alpha }(q_{n})}} 
{e^{(1 - \alpha ) H_{\mu }^{\alpha }(q_{n})}} =   \frac{ \int_{I}  
g^{\beta_{1}} \, d \lambda  }{ \int_{\mathbb{R}}  g^{\beta_{1}}\, d \lambda }\mu ( I )^{- \alpha }.
\end{equation}
It is easy to see that (\ref{propentrloc0}) reduces to the traditional
point density result (\ref{eq_pointdensbucklew}) for  $\alpha=0$. 

In Corollary~\ref{corollary123} we present an almost immediate
consequence of (\ref{propentrloc0}) and 
(\ref{eq_condasympopt}) which  concerns the
distortion contribution of an arbitrary finite interval $I$:
\[
 \lim_{n\to\infty} \frac{\int_{I}|x-q_{n}(x)|^{r}\, \mu(dx)}{
   D_{\mu}(q_n) }=
  \frac{ \int_{I}  
g^{\beta_{1}} \, d \lambda  }{ \int_{\mathbb{R}}  g^{\beta_{1}}\, d
    \lambda } .
\]
Thus the probability density  $g^{\beta_1}/ \int g^{\beta_1} \,
d\lambda $ can be interpreted as either the (R\'enyi) \emph{entropy density} or
the \emph{distortion density} of any asymptotically optimal quantizer
sequence  $(q_n)$.

\subsection{Mismatch}

\label{subsec_mismatch}
For scalar quantization Bucklew's fixed-rate mismatch result
\cite[Thm.\ 2]{Buc84} can be stated as follows: If a sequence of $n$-level
quantizers  $(q_n)$ that is  asymptotically optimal for a source
with distribution $\mu= g\lambda$ is applied  to a source with
distribution $\nu=f\lambda$, then (under some assumptions on $g$ and $f$)
\[
\lim_{n\to \infty} n^r D_{\nu}(q_n) =   \frac{1}{(1+r)2^r} \int
\frac{f}{g_*^r} \, d\lambda
\]
where $g_* = g^{\frac{1}{1+r}}/ \int g^{\frac{1}{1+r}} \, d\lambda $
is the optimal point density function for $\mu=g\lambda$ from
(\ref{eq_pointdensbucklew}). This is a generalization of a classical
result of Bennett \cite{Ben48} who considered  companding quantization and mean
square distortion. The integral on the right hand side is often called
Bennett's integral.  In view of (\ref{eq_pointdensbucklew}),
and after some calculations, we obtain  that the asymptotic
performance loss due to mismatch is \vspace{0.4cm}
\begin{equation}
\label{eq_0loss}
  \lim_{n\to \infty} \frac{ D_{\nu}(q_n)}{ D^{0}_{\nu}(\log n)} 
=  e^{r\,  \mathcal{D}_{1+r}(f_*\|g_*)}
\end{equation} 
where  $f_* =  f^{\frac{1}{1+r}}/ \int f^{\frac{1}{1+r}} \, d\lambda $
is the optimal point density for $\nu=f\lambda$ and
\begin{equation}
\label{renyi-div}
\mathcal{D}_{\alpha}(u\|v) = \frac{1}{\alpha-1} \log\biggl( \int
u^{\alpha } v^{1-\alpha} \, d\lambda \biggr)
\end{equation}
denotes the R\'enyi divergence
of order $\alpha\neq 1$ between densities $u$ and $v$. (Thus the loss is
always greater than one  unless $\mu=\nu$).

For the entropy-constrained case the main result of \cite{GrLi03} implies that if $(q_n)$ is
asymptotically optimal for $\mu=g\lambda$, but it is used for
$\nu=f\lambda$, then 
\[
  \lim_{n\to \infty} e^{rH^1_{\nu}(q_n)} D_{\nu}(q_n)=
 \frac{1}{(1+r)2^r}  e^{r  h^{1}(f)}  e^{r\,  \mathcal{D}_1(f\|g)}.
\]
Here  $ \mathcal{D}_1(f\|g)=\mathcal{D}(f\|g)= \int f \log  \frac{f}{g}\,
  d\lambda $ is the Kullback-Leibler divergence (relative entropy)
  between $f$ and $g$. From (\ref{eq_pointdensbucklew}) the loss due
  to mismatch is
\begin{equation}
\label{eq_1loss}
  \lim_{n\to \infty} \frac{ D_{\nu}(q_n)}{ D^{1}_{\nu}(H^1_{\nu}(q_n))} 
=  e^{r\,  \mathcal{D}_{1}(f\|g)}.
\end{equation}

In Theorem~\ref{mresul} we present a result on mismatch for quantization with
constrained R\'enyi entropy of order $\alpha\in (0,1)$. The result
states that if
$(q_n)$ is asymptotically optimal for $\mu=g\lambda$, but is applied
to $\nu=f\lambda$, then
\[
  \lim_{n\to \infty} e^{r H^{\alpha}_{\nu}(q_n)} D_{\nu}(q_n)=
 \frac{1}{(1+r)2^r}  e^{-r
   \mathcal{D}_{\alpha}(f\|g_{\alpha,r})} \int
 \frac{f}{\bigl(g_{\alpha,r}\bigr)^r} \, d\lambda 
\]
where 
\begin{equation}
\label{g_alpha_r}
g_{\alpha,r} =\frac{g^{\frac{1}{\beta_2}}}{\int
    g^{\frac{1}{\beta_2}}\, d\lambda}
\end{equation}
 with  $\beta_{2} =
\frac{1-\alpha + r}{1 - \alpha}$ (note that $g_{0,r}=g_*$). 
The loss due to mismatch can be expressed as 
\begin{equation}
\label{eq_alphaloss}
  \lim_{n\to \infty} \frac{ D_{\nu}(q_n)}{ D^{\alpha}_{\nu}(H^{\alpha}_{\nu}(q_n))} 
=   \frac{e^{r(
   \mathcal{D}_{1+r}(f_{0,r}\|g_{\alpha,r}) -
   \mathcal{D}_{\alpha}(f\|g_{\alpha,r}))}}{ e^{r(
   \mathcal{D}_{1+r}(f_{0,r}\|f_{\alpha,r}) -
   \mathcal{D}_{\alpha}(f\|f_{\alpha,r}))}}. 
\end{equation}
The loss can be seen to be always greater than one unless
$\mu=\nu$ (see Remark~\ref{rem_mismatch} following Theorem~\ref{mresul}). Setting formally $\alpha=0$
or $\alpha=1$ (or, more 
precisely, letting $\alpha\downarrow 0$ or $\alpha\uparrow 1$) in the
above formula yields the known cases (\ref{eq_0loss}) and~(\ref{eq_1loss}).

\section{Entropy density and related results} 

\label{sec_density}

Throughout this section we assume that $\mu=g\lambda$.
For $r\ge 1$ and $\alpha \in [0, 1+r)\setminus \{1\}$ let 
\begin{equation}
\label{defpar1}
\beta_{1} =
\frac{1-\alpha + \alpha r}{1 - \alpha + r}, \qquad
\beta_{2} =
\frac{1-\alpha + r}{1 - \alpha}.
\end{equation}

\begin{definition}
Let $ C(r)=\frac{1}{2^{r}(1+r)}$ and
define, for $\alpha \in [0, 1+r) \setminus \{ 1 \}$, 
\[
Q_{\alpha, r}(\mu ) = 
C(r)\left( \int g^{\beta_{1}}\, d \lambda\right)^{\beta_{2}}
\label{defqcoeff}
\]
whenever the integral is finite. Note that $Q_{\alpha, r}(\mu)
\in {} (0, \infty )$.
We call 
$Q_{\alpha, r}(\mu)$ the \emph{quantization coefficient} of $\mu$. 
\end{definition}

\begin{definition}
\label{def_weak_unimod}
A one-dimensional probability density function $g$ is called \emph{weakly
unimodal} if it is continuous on its support and there exists an
$l_{0}>0$ such that $\{x: g(x)\ge l\}$ is a compact interval for every
$l \in (0, l_{0})$.
\end{definition}

Note that every weakly unimodal density is bounded and its support is
a (possibly unbounded) interval. Clearly, all continuous unimodal
densities are weakly unimodal. The class of weakly unimodal
densities includes many  parametric source density classes commonly
used in modeling information sources such as exponential, Laplacian,
Gaussian, generalized Gaussian, and all bounded 
gamma and beta densities.

 The following is one of
the main results in \cite{KrLi11}.

\begin{theorem}[{\cite[Thm~3.4]{KrLi11}}]
\label{thmkrlmain}  For $r>1$ and $\alpha \in
  (0,1)$, if $\mu$ has a weakly unimodal density $g$ and $\int
  |x|^{r+\delta}\, d \mu(x) < \infty$ for some $\delta > 0$, then
  $Q_{\alpha, r}(\mu)$ is well defined and
\[
\lim_{R \rightarrow \infty }e^{rR} D_{\mu}^{\alpha}(R) = Q_{\alpha, r}(\mu) .
\]
\end{theorem}

\begin{remark} (a) The theorem  and   (\ref{eq_mainasympt}) express the
same asymptotic result since 
\[
\left( \int g^{\beta_{1}}\, d \lambda\right)^{\beta_{2}}= e^{r
  h^{\beta_1}(g)}.
\]
The quantization coefficient $Q_{\alpha,r}$  can also be expressed in terms of R\'enyi
divergences \eqref{renyi-div} and the density 
$g_{\alpha,r}$    introduced in \eqref{g_alpha_r}. One can easily
verify that 
\begin{equation}
\label{eq_auxxxx}  
\left( \int g^{\beta_{1}}\, d \lambda\right)^{\beta_{2}} =
 e^{-r
   \mathcal{D}_{\alpha}(g\|g_{\alpha,r})} \int
 \frac{g}{\bigl(g_{\alpha,r}\bigr)^r} \, d\lambda.
\end{equation}
Furthermore, for any density $h$ with $ \int \frac{g}{h^r} \,
d\lambda<\infty$, 
\begin{eqnarray*}
 \int \frac{g}{h^r} \, d\lambda &=&  \left(\int g^{\frac{1}{1+r}}\,
 d\lambda \right)^{1+r} \int \left(\frac{g^{\frac{1}{1+r}}}{\int g^{\frac{1}{1+r}}\,
 d\lambda}\right)^{1+r}  h^{(1-(1+r))}\, d\lambda    \\
&=& \|g\|_{\frac{1}{1+r}} \int (g_{0,r})^{1+r} h^{(1-(1+r))}\,
 d\lambda \\
&=& \|g\|_{\frac{1}{1+r}}\, e^{r\,  \mathcal{D}_{1+r}(g_{0,r}\|h)}.
\end{eqnarray*}  
Substituting $h=g_{\alpha,r}$ and combining with \eqref{eq_auxxxx} we obtain
\[
\left( \int g^{\beta_{1}}\, d \lambda\right)^{\beta_{2}} =
 \|g\|_{\frac{1}{1+r}} e^{r( \mathcal{D}_{1+r}(g_{0,r}\|g_{\alpha,r}) -
   \mathcal{D}_{\alpha}(g\|g_{\alpha,r}))}.
\]

\noindent (b)  Theorem~3.4 in \cite{KrLi11} also covers the more
exotic  $\alpha\in [-\infty,0)$  case, but for technical reasons we
  require that   $\alpha\in (0,1)$. The weak unimodality condition is a
  technical one and most likely can be significantly relaxed.

\end{remark}

\begin{definition}
\label{defasymp}
A sequence of quantizer $(q_{n})_{n \in \mathbb{N}}$ with $H_{\mu
}^{\alpha }(q_{n}) \to \infty$ as $n \to \infty$ is called
$\alpha$-asymptotically optimal for $\mu$ if
\[
\lim_{n \rightarrow \infty } \frac{D_{\mu }(q_{n})}{ D_{\mu }^{\alpha } ( H_{\mu }^{\alpha }(q_{n}) ) } = 1.
\]
\end{definition}

\begin{remark} In what follows  we will simply write ``asymptotically
  optimal'' instead of ``$\alpha$-asymptotically optimal.''    Under the conditions of Theorem~\ref{thmkrlmain}, a
  quantizer sequence  $(q_n)$ with $H_{\mu 
}^{\alpha }(q_{n}) \to \infty$  is asymptotically
optimal for $\mu$ if and only if
\[
\lim_{n \rightarrow \infty }e^{rH_{\mu}^{\alpha }(q_{n})} D_{\mu }(q_{n}) = Q_{\alpha, r}(\mu) .
\]
\end{remark}

\smallskip

For any measurable $A \subset \mathbb{R}$ with $\mu (A) > 0$ we denote
by $\mu ( \cdot | A )$ the conditional probability for $\mu$ given
$A$.  Let $c,d\in \R$ be such that $c<d$ and $\mu((c,d])\in (0,1)$,
but otherwise arbitrary.  In the following theorem, we let
$A_{1}=(c,d]$, $A_{2}=\mathbb{R} \setminus A_{1}$, and $\mu_i=
\mu(\cdot|A_i)$ for $i\in \{1,2\}$.

\begin{theorem}
\label{theoremdens}
Let $r>1$ and $\alpha \in (0,1)$. Let  $\mu=g\lambda$,  where the
density function $g$  is weakly unimodal and satisfies  $\int
|x|^{r+\delta}\, d \mu(x) < \infty$ for some $\delta > 0$.  Let
$(q_{n})_{n \in \mathbb{N}}$ be an asymptotically optimal sequence for
$\mu$.  Then, 
 for $i \in \{ 1,2 \}$,
\begin{equation}
\label{propentrloc}
\lim_{n\to \infty} \frac{e^{(1 - \alpha ) H_{\mu_i}^{\alpha }(q_{n})}} 
{e^{(1 - \alpha ) H_{\mu }^{\alpha }(q_{n})}} =   \frac{ \int_{A_{i}}  
g^{\beta_{1}} \, d \lambda  }{ \int_{\mathbb{R}}  g^{\beta_{1}}\, d \lambda }\mu ( A_{i} )^{- \alpha }.
\end{equation}
and $(q_n)$ is asymptotically optimal for $\mu_i$, i.e., $\lim_{n\to
  \infty}H_{\mu_i}^{\alpha }(q_{n}) = \infty$ and 
\begin{equation}
\label{cutsequopt}
\lim_{n\to \infty} e^{r H_{\mu_i}^{\alpha }(q_{n})}
D_{\mu_i}(q_{n})  
= Q_{\alpha , r}( \mu_i ).
\end{equation}

\end{theorem}

\begin{remark} (a) \ As discussed in Section~\ref{sec_entrdist}, the
  ratio on the left hand side of \eqref{propentrloc} can be
  interpreted as the relative contribution to R\'enyi entropy of
  interval $I$. The theorem determines the limit of this relative
  entropy contribution for a sequence of asymptotically optimal
  quantizers. The method used in the proof is a generalization of the
  approach developed by Bucklew \cite{Buc84} for the case $\alpha =
  0$.

\smallskip

\noindent (b)\ Using $\alpha \in ( 0,1 )$ and the condition $\int
|x|^{r+\delta} \, d \mu(x) < \infty$, the integral in the definition
of $Q_{\alpha , r}( \mu )$ can be shown to be finite by an application
of H\"older's inequality as in \cite[Remark 6.3 (a)]{GrLu00}. For the
same reason, $Q_{\alpha , r}(\mu_i)$ is finite for $i\in \{1,2\}$.

\end{remark}

In the  proof of the theorem we will need the following lemma which
is proved in the Appendix.

\begin{lemma} 
\label{lemmaux}
Under the conditions of \ Theorem~\ref{theoremdens} the
  following hold: 
For $i \in \{ 1,2 \}$,
 \begin{equation}
\label{exqkoeff}
\lim_{n\to \infty}H_{\mu_i}^{\alpha }(q_{n}) = \infty
\end{equation}
and for all $p\in \R$, 
\begin{equation}
\label{entrconvzero}
\lim_{n\to \infty} \frac{\mu ( q_{n}^{-1}(q_{n}(p)) )^{\alpha }}{ \sum_{a \in q_{n}(\mathbb{R})} \mu ( q_{n}^{-1}(a) )^{\alpha }  } 
= 0, \qquad 
\lim_{n\to \infty} \frac{\mu_i ( q_{n}^{-1}(q_{n}(p)))^{\alpha }}{ \sum_{a \in q_{n}(\mathbb{R})}\mu_i (  q_{n}^{-1}(a)  )^{\alpha }  } 
= 0.
\end{equation}

\end{lemma}

\noindent\emph{Proof of Theorem~\ref{theoremdens}.} \ 
We begin the
proof by showing that (\ref{propentrloc}) holds if we 
additionally assume that for $i\in \{1,2\}$, 
\begin{equation}
\label{noninfass}
\limsup_{n \to \infty} e^{r ( H_{\mu}^{\alpha }(q_{n}) - H_{\mu_i}^{\alpha }(q_{n})  ) } < \infty .
\end{equation}
In this case, any subsequence of $(q_n)$ has a sub-subsequence, which we also
denote by $(q_n)$, such that 
\begin{equation}
\label{fghztr}
\lim_{n\to \infty} e^{r ( H_{\mu}^{\alpha }(q_{n}) - H_{\mu_1}^{\alpha
  }(q_{n})  ) } = d^{\frac{r}{1 - \alpha }}
\end{equation}
for some $d\in [0,\infty)$. The obvious bound
\begin{equation}
\label{dlowerbound}
e^{r ( H_{\mu}^{\alpha }(q_{n}) - 
H_{\mu_1}^{\alpha }(q_{n}) ) }  \ge \mu ( A_{1} )^{\frac{\alpha r}{1
      - \alpha}}
\end{equation}
implies that $d>0$. In what follows we show that $d$ is independent of
the choice of the sub-subsequence (and thus the limit
in (\ref{fghztr}) holds for the original sequence) and
explicitly identify $d$. 

For any two sequences $(u_{n})$ and $(v_{n})$ of positive reals we
write $u_{n} \sim v_{n}$ if 
\begin{equation}
\label{simdef}
\lim_{n \to \infty} \frac{u_{n}}{v_{n}} = 1.
\end{equation}
Note that if $u_n\sim v_n$ and $u'_n\sim v'_n$, then $(u_n+
u'_n)\sim (v_n+v'_n)$ and $u_n \cdot u'_n\sim v_n\cdot v'_n$.
We can rewrite  (\ref{fghztr})  as 
\begin{equation}
\label{simentro}
e^{(1- \alpha ) H_{\mu_1}^{\alpha }(q_{n}) } 
\sim \frac{1}{d} e^{(1-\alpha ) H_{\mu }^{\alpha }(q_{n}) } .
\end{equation}
We note that
\begin{eqnarray}
e^{r H_{\mu}^{\alpha }(q_{n})} D_{\mu }(q_{n}) 
&=&
e^{r H_{\mu}^{\alpha }(q_{n})} \sum_{i=1}^{2} \mu ( A_{i} ) D_{\mu_i}(q_{n}) \nonumber \\
&=& 
\sum_{i=1}^{2} e^{r ( H_{\mu}^{\alpha }(q_{n}) - H_{\mu_i}^{\alpha }(q_{n})  ) }  \mu ( A_{i} )
e^{r H_{\mu_i}^{\alpha }(q_{n})} D_{\mu_i}(q_{n}).
\label{parterrtwo}
\end{eqnarray}
Since the cells of $q_n$ are intervals, at most  two of them   may 
intersect both $A_1=(c,d]$ and $A_2=\R\setminus (c,d]$ (namely,
those containing $c$ and $d$). Then (\ref{entrconvzero})  implies 
\begin{eqnarray*}
e^{r ( H_{\mu}^{\alpha }(q_{n}) - H_{\mu_2}^{\alpha }(q_{n})  ) }
&=&
e^{r  H_{\mu}^{\alpha }(q_{n})} \left( e^{(1 - \alpha ) 
H_{\mu_2}^{\alpha }(q_{n})   } \right)^{-\frac{r}{1 - \alpha }} \\*
&\sim & 
e^{r  H_{\mu}^{\alpha }(q_{n})} \mu ( A_{2} )^{ \frac{ \alpha r }{1 - \alpha } }
\left( \sum_{a\in q_n(\R):\, q_{n}^{-1}(a) \subset A_{2} } 
\mu ( q_{n}^{-1}(a) )^{\alpha }   \right)^{-\frac{r}{1 - \alpha }} \\
&\sim & 
e^{r  H_{\mu}^{\alpha }(q_{n})} \mu ( A_{2} )^{ \frac{ \alpha r }{1 - \alpha } }
\left( e^{(1 - \alpha )H_{\mu }^{\alpha }(q_{n})}  -  
\sum_{a\in q_n(\R):\,  q_{n}^{-1}(a) \subset A_{1} } 
\mu ( q_{n}^{-1}(a) )^{\alpha }   \right)^{-\frac{r}{1 - \alpha }} \\
&\sim & 
e^{r  H_{\mu}^{\alpha }(q_{n})} \mu ( A_{2} )^{ \frac{ \alpha r }{1 - \alpha } }
\left( e^{(1 - \alpha )H_{\mu }^{\alpha }(q_{n})}  -  
e^{(1 - \alpha )H_{\mu_1}^{\alpha }(q_{n})} \mu ( A_{1} )^{\alpha }
\right)^{-\frac{r}{1 - \alpha }} .
\end{eqnarray*}
In view of  (\ref{simentro}) we conclude
\begin{eqnarray}
e^{r ( H_{\mu}^{\alpha }(q_{n}) - H_{\mu_2}^{\alpha }(q_{n})  ) } 
& \sim & 
e^{r  H_{\mu}^{\alpha }(q_{n})} \mu ( A_{2} )^{ \frac{ \alpha r }{1 - \alpha } }
\left( e^{(1 - \alpha )H_{\mu }^{\alpha }(q_{n})}  -  
\frac{1}{d} e^{(1-\alpha ) H_{\mu }^{\alpha }(q_{n}) } \mu ( A_{1} )^{\alpha }
\right)^{-\frac{r}{1 - \alpha }} \nonumber \\
& = &
\mu ( A_{2} )^{ \frac{ \alpha r }{1 - \alpha } }
\left( 1  -  
\frac{1}{d}  \mu ( A_{1} )^{\alpha }
\right)^{-\frac{r}{1 - \alpha }}.
\label{erhmu}
\end{eqnarray}
Applying (\ref{erhmu}) and (\ref{fghztr}) to (\ref{parterrtwo}) we obtain
\begin{eqnarray}
Q_{\alpha , r} (\mu ) \nonumber 
&\sim & e^{r H_{\mu}^{\alpha }(q_{n})} D_{\mu }(q_{n}) \nonumber \\
& = &
e^{r ( H_{\mu}^{\alpha }(q_{n}) - H_{\mu_1}^{\alpha }(q_{n})  ) }  \mu ( A_{1} )
e^{r H_{\mu_1}^{\alpha }(q_{n})} D_{\mu_1}(q_{n}) \nonumber \\
&&  \mbox{}  + 
e^{r ( H_{\mu}^{\alpha }(q_{n}) - H_{\mu_2}^{\alpha }(q_{n})  ) }  \mu ( A_{2} )
e^{r H_{\mu_2}^{\alpha }(q_{n})} D_{\mu_2}(q_{n}) \nonumber \\
& \sim &
d^{\frac{r}{1 - \alpha }}  \mu ( A_{1} )
e^{r H_{\mu_1}^{\alpha }(q_{n})} D_{\mu_1}(q_{n}) \nonumber \\
&& \mbox{}   + 
\mu ( A_{2} )^{ \frac{ \alpha r }{1 - \alpha } }
\left( 1  -  
\frac{1}{d}  \mu ( A_{1} )^{\alpha }
\right)^{-\frac{r}{1 - \alpha }}
\mu ( A_{2} )
e^{r H_{\mu_2}^{\alpha }(q_{n})} D_{\mu_2}(q_{n}) \nonumber \\
&=&
\mu ( A_{1} )^{\beta_{1} \beta_{2}}
\left(  \mu ( A_{1} )^{- \alpha } d\right)^{\frac{r}{1 - \alpha }}  
e^{r H_{\mu_1}^{\alpha }(q_{n})} D_{\mu_1}(q_{n}) \nonumber \\
&&\mbox{}  + 
\mu ( A_{2} )^{\beta_{1} \beta_{2}}
\left( \frac{1}{1 - \frac{1}{  \mu ( A_{1} )^{- \alpha }d  }} \right)^{\frac{r}{1 - \alpha }}
e^{r H_{\mu_2}^{\alpha }(q_{n})} D_{\mu_2}(q_{n}). 
\label{qkopartpart}
\end{eqnarray}
Since $H_{\mu_i}^{\alpha }(q_{n}) \to \infty$ by  (\ref{exqkoeff}),
Theorem~\ref{thmkrlmain}  implies\footnote{Strictly speaking,
  Theorem~\ref{thmkrlmain}  (\cite[Thm~3.4]{KrLi11}) does not apply
  for $\mu_2$ since its 
  density $g_2$ is not weakly unimodal. However, $g_2$ is the mixture
  of two weakly unimodal densities with well-separated supports, and the
  proof of \cite[Thm~3.4]{KrLi11} can easily be extended to this case.}
\begin{equation}
\label{dilower}
\liminf_{n\to \infty} e^{r H_{\mu_i}^{\alpha }(q_{n})}
D_{\mu_i}(q_{n}) \ge Q_{\alpha , r}(\mu_i), \quad i\in \{1,2\}
\end{equation}
and thus the limit inferior of the  the right hand  
side of (\ref{qkopartpart}) is lower bounded by 
\begin{eqnarray}
\lefteqn{ \left(  \mu ( A_{1} )^{- \alpha }d \right)^{\frac{r}{1 - \alpha }}  
Q_{\alpha , r}(\mu_1) \mu ( A_{1} )^{\beta_{1} \beta_{2}} 
+ 
\left( \frac{1}{1 - \frac{1}{  \mu ( A_{1} )^{- \alpha } d }} \right)^{\frac{r}{1 - \alpha }}
Q_{\alpha , r}(\mu_2) \mu ( A_{2} )^{\beta_{1} \beta_{2}} }  \nonumber \\ 
& =& C(r)
\left( \mu ( A_{1} )^{- \alpha } d\right)^{\frac{r}{1 - \alpha }}
\left( \int_{A_{1}} g^{\beta_{1}}\, d \lambda \right)^{\beta_{2}} 
+  C(r)
\left( \frac{1}{1 - \frac{1}{  \mu ( A_{1} )^{- \alpha } d }} \right)^{\frac{r}{1 - \alpha }}
\left( \int_{A_{2}} g^{\beta_{1}}\, d \lambda \right)^{\beta_{2}}.
\label{lowbqko}
\end{eqnarray}
In view of the definition of $Q_{\alpha,r}(\mu)$, combining
(\ref{qkopartpart}) and (\ref{lowbqko})  yields
\begin{equation}
\label{lowbo}
\left( \int g^{\beta_{1}} \, d \lambda \right)^{\beta_{2}}
\geq F(d)
\end{equation}
where 
\[
F(d)= \left( \mu ( A_{1} )^{- \alpha }d \right)^{\frac{r}{1 - \alpha }}
\left( \int_{A_{1}} g^{\beta_{1}}\, d \lambda \right)^{\beta_{2}} 
+  
\left( \frac{1}{1 - \frac{1}{  \mu ( A_{1} )^{- \alpha }d  }} \right)^{\frac{r}{1 - \alpha }}
\left( \int_{A_{2}} g^{\beta_{1}}\, d \lambda \right)^{\beta_{2}}.
\]
Now let 
\begin{equation}
\label{infford}
d_{0} = \mu ( A_{1} )^{\alpha } 
\frac{ \int g^{\beta_{1}}\, d \lambda  }{ \int_{A_{1}} g^{\beta_{1}}\, d \lambda  } 
\end{equation}
and note that the bound  (\ref{dlowerbound}) implies  $d^{-1} \mu (
A_{1} )^{ \alpha } \in (0,1]$. Moreover, from (\ref{noninfass}) we
  actually obtain   $d^{-1} \mu (
A_{1} )^{ \alpha } \in (0,1)$.
Thus if $d\neq d_0$, then  
Lemma \ref{lowbouasymp} in the Appendix gives  $F(d) > F(d_{0})$. Moreover, a simple calculation yields
$F(d_{0})=\left( \int g^{\beta_{1}}\, d \lambda \right)^{\beta_{2}}$. Hence we deduce from (\ref{lowbo}) that
$d=d_{0}$. Because we  chose an arbitrary convergent subsequence 
in (\ref{fghztr}), we obtain that  (\ref{fghztr}) actually holds with
$d=d_0$ for the 
original quantizer sequence. This and  (\ref{infford}) 
yield (\ref{propentrloc}) for $i=1$. Also, (\ref{erhmu}) and
(\ref{infford}) imply  (\ref{propentrloc}) for $i=2$. 

\smallskip 

As next step we will prove that (\ref{cutsequopt}) is true 
under the assumption (\ref{noninfass}).
We proceed indirectly. Assume first that 
(\ref{cutsequopt}) is not true for $i=1$. Then by (\ref{dilower}) we
can choose a 
subsequence of $(q_{n})$, also denoted by $(q_{n})$, 
such that
\[
\lim_{n\to \infty} e^{ r H_{\mu_1}^{\alpha }(q_{n})  }
D_{ \mu_1}(q_{n})  >
Q_{\alpha , r}(\mu_1)= C(r) \left( \int_{A_{1}} 
\left( \frac{g}{\mu ( A_{1} )} \right)^{\beta_{1} }\, d \lambda  \right)^{\beta_{2}} .
\]
We deduce from (\ref{parterrtwo}) and  (\ref{propentrloc}) that
\begin{eqnarray}
\limsup_{n \to \infty } e^{ r H_{\mu_2}^{\alpha }(q_{n})  }
D_{ \mu_2}(q_{n}) 
&<&
C(r) \left( \int_{A_{2}} 
\left( \frac{g}{\mu ( A_{2} )} \right)^{\beta_{1} }\, d \lambda  \right)^{\beta_{2}}
\label{striequerr}
\end{eqnarray}
since otherwise we would have
\begin{eqnarray*}
\limsup_{n\to \infty}  e^{ r H_{\mu}^{\alpha }(q_{n})  }
D_{ \mu}(q_{n}) & > & C(r) \sum_{i=1}^2  \left( \frac{ \int   
g^{\beta_{1}} \, d \lambda  }{ \int_{A_i}  g^{\beta_{1}}\, d
  \lambda }\right)^{\beta_2-1}  \left( \int_{A_{i}} 
g^{\beta_{1}}\, d \lambda
\right)^{\beta_{2}}  \\
&=& Q_{\alpha,r}(\mu)
\end{eqnarray*}
which would contradict the asymptotic optimality of $(q_n)$.
But the right hand side of (\ref{striequerr}) is
$Q_{\alpha,r}(\mu_2)$, which contradicts (\ref{dilower}), so
(\ref{cutsequopt}) must hold for $i=1$. Similarly, we end in a contradiction if we assume
that (\ref{cutsequopt}) does not hold for $i=2$. 

\smallskip 

It remains to prove that (\ref{noninfass}) must
hold. Assuming the contrary, we have 
\[
\liminf_{n\to \infty} e^{r ( H_{\mu_i}^{\alpha }(q_{n}) - H_{\mu
  }^{\alpha }(q_{n}) ) }  =0.
\]
Since $D_{\mu}(q_n)\ge \mu(A_i)D_{\mu_i}(q_n)$, 
\begin{eqnarray*}
0 &=& \liminf_{n \to \infty} e^{r ( H_{\mu_i}^{\alpha }(q_{n}) - H_{\mu }^{\alpha }(q_{n}) ) } \\
&=&
\liminf_{n \to \infty} 
\frac{ e^{r H_{\mu_i}^{\alpha }(q_{n})} 
\frac{D_{\mu }(q_{n})}{D_{\mu_i}(q_{n})} 
D_{\mu_i}(q_{n})}
{ e^{r H_{\mu }^{\alpha }(q_{n}) )} D_{\mu }(q_{n}) } \\
& \geq &
\frac{\mu ( A_{i} )}{Q_{\alpha , r}(\mu )} \,
\liminf_{n \to \infty} e^{r H_{\mu_i}^{\alpha }(q_{n})} D_{\mu_i}(q_{n}) ,
\end{eqnarray*}
which would imply
\[
\liminf_{n\to \infty} e^{r H_{\mu_i}^{\alpha }(q_{n})} D_{\mu_i}(q_{n})= 0
\]
contradicting (\ref{dilower}).
Hence  (\ref{noninfass}) must hold and the proof is complete. \qed

%% \begin{remark}
%% If $\alpha = 0$, then the right hand side of (\ref{propentrloc}) is
%% identical to the point density measure of $A_{1}$ (see
%% e.g. \cite[Theorem 7.5]{GrLu00}).
%% \end{remark}

\medskip

Let $(q_{n})_{n \in \mathbb{N}} $ be a sequence of quantizers and for
any $n \ge 1$ and any Borel set $E \subset \mathbb{R}$  define
\begin{equation}
\label{defmgn}
M_{g}^{n}(E) = e^{r H_{\mu }^{\alpha }(q_{n})} \int_{E}
|x-q_{n}(x)|^{r} g(x) \, d \lambda(x) .
\end{equation}
Moreover, for  $\alpha \in [0, 1+r) \setminus \{ 1 \}$ let
\begin{equation}
\label{mgedef}
M_{g}(E) = C(r) \left( \int_{E} g^{\beta_{1}}\, d \lambda \right) \left( \int_{\mathbb{R}} 
g^{\beta_{1}} d \lambda \right)^{\frac{r}{1 - \alpha }} .
\end{equation}
Clearly, $M_{g}^{n}$ and $M_{g}$ are Borel-measures on
$\mathbb{R}$ that are absolutely continuous with respect to  $\lambda$. 
We define the probability measure $\hat{\mu}$ by setting, for any
Borel set $E\subset \R$, 
\begin{equation}
\label{defmuhat}
\hat{\mu}(E) = \frac{\int_Eg^{\beta_{1}}\, d\lambda }{\int_{\R}
  g^{\beta_{1}}\, d \lambda }.
\end{equation}

\begin{corollary}
\label{corollary123}
Let $r>1$ and $\alpha \in (0, 1)$. 
Suppose  that  $\mu=g\lambda$,  where the
density function $g$  is weakly unimodal and satisfies  $\int
|x|^{r+\delta}\, d \mu(x) < \infty$ for some $\delta > 0$.
If $(q_{n})_{n \in \mathbb{N}}$ is 
an asymptotically optimal sequence of quantizers for $\mu$, then for
any $c,d \in \mathbb{R}$ such that $-\infty < c < d < \infty$ we have 
\begin{itemize}
\item[(i)]
$\displaystyle \lim\limits_{n\to\infty} \frac{\int_{(c,d]}|x-q_{n}(x)|^{r}g(x)\, d \lambda(x)}{
  \int_{\mathbb{R}}|x-q_{n}(x)|^{r}g(x)\, d \lambda(x) }= \hat{\mu}((c,d])$;
\item[(ii)] $M_{g}^{n}$ converges weakly to $M_{g}$.
\end{itemize}
\end{corollary}

\begin{remark}
\label{rem_coinc}
Combining Theorem~\ref{theoremdens}  and the corollary and using
the $\sim$ notation introduced in (\ref{simdef}), we observe that 
\begin{equation}
\label{intcd}
 \frac{\int_{(c,d]}| x - q_{n}(x) |^{r}  d \mu (x) }
{ \int_{\mathbb{R}}  | x - q_{n}(x) |^{r} d \mu (x) }
\sim 
\frac{\sum_{ a \in q_{n}( \mathbb{R} ) } \mu ( q_{n}^{-1}(a) \cap (c,d] )^{\alpha } }
{\sum_{ a \in q_{n}( \mathbb{R} ) } \mu ( q_{n}^{-1}(a) )^{\alpha } }.
\end{equation}
This means that the relative error and entropy contributions of $(q_n)$
over any given interval asymptotically   coincide.
\end{remark}

\noindent \emph{Proof of Corollary~\ref{corollary123}.}
We start by proving (i). 
Let $A=(c,d]$ and define 
\[
\mu_{n}(A) = \frac{\int_{A}|x-q_{n}(x)|^{r}g(x)\, d \lambda(x)}{
  \int_{\mathbb{R}}|x-q_{n}(x)|^{r}g(x)\, d \lambda(x) }.
\]
Obviously we can assume without loss of generality that $\mu (A) \in (0,1)$.
Applying  (\ref{propentrloc}) and  (\ref{cutsequopt}) in Theorem~\ref{theoremdens},   we obtain 
\begin{eqnarray}
\mu_{n} (A) &=& 
\frac{ e^{r H_{\mu }^{\alpha }(q_{n}) } 
\int_{A} |x - q_{n}(x)|^{r} g(x)\,  d \lambda(x) }
{ e^{r H_{\mu }^{\alpha }(q_{n}) } 
\int |x - q_{n}(x)|^{r} g(x)\,  d \lambda(x) } \nonumber \\
& \sim & 
\frac{ e^{r ( H_{\mu }^{\alpha }(q_{n}) - 
H_{\mu ( \cdot | A ) }^{\alpha }(q_{n}) ) } 
 \mu (A) \,e^{r H_{\mu ( \cdot | A ) }^{\alpha }(q_{n}) }
  \int_{A} |x - q_{n}(x)|^{r} \frac{ g(x) }{\mu (A)}\,  d \lambda(x)} 
{ Q_{\alpha , r }( \mu ) } \nonumber \\
& \sim & 
\frac{ e^{r ( H_{\mu }^{\alpha }(q_{n}) - H_{\mu ( \cdot | A ) }^{\alpha }(q_{n}) ) } 
Q_{ \alpha, r }( \mu ( \cdot | A ) ) \mu (A) } 
{ Q_{\alpha , r }( \mu ) } \nonumber \\
& \sim &
\left( \frac{\mu ( A )^{\alpha} \int g^{\beta_{1}  } 
\, d \lambda}{ \int_{A} g^{\beta_{1}  }\, d \lambda} \right)^{\frac{r}{1 - \alpha }}
\frac{ 
Q_{ \alpha, r }( \mu ( \cdot | A ) ) \mu (A) } 
{  Q_{\alpha , r }( \mu )}.
\label{rhsqkoeff}
\end{eqnarray}
Definition~\ref{defqcoeff},  \eqref{defmuhat},  and  a straightforward
calculation yield 
that the right hand side of (\ref{rhsqkoeff}) is equal to $\hat{ \mu }(A)$.

Next we prove (ii). Because $(q_{n})$ is asymptotically optimal for
$\mu$ we have $M_{g}^{n}(\mathbb{R}) \to M_{g}(\mathbb{R})$ as $n \to
\infty$. Moreover, $M_{g}$ is a finite measure.  Due to a refined
version of the Portmanteau theorem \cite[Thm.\ 2.4 and Example
  2.3]{Bil99} it suffices to prove that $M_{g}^{n}((c,d]) \to
  M_{g}((c,d])$ for any $-\infty < c< d<\infty$. Let $A=(c,d]$ and
      assume $\mu(A)>0$, since otherwise
      $M_{g}^{n}(A)=M_{g}(A)=0$ for all $n$. Applying the definitions
      (\ref{defmgn}) and (\ref{mgedef}), we obtain
\[
\frac{M_{g}^{n}(A)}{M_{g}(A)} =
\frac{\int_{A} |x-q_{n}(x)|^{r} g(x)\,  d \lambda(x)}{D_{\mu
  }(q_{n})}\cdot 
\frac{e^{r H_{\mu }^{\alpha }(q_{n})} D_{\mu }(q_{n})}{C(r) \left( \int g^{\beta_{1} } 
\, d \lambda \right)^{\frac{r}{1 - \alpha }}  \int_{A}
g^{\beta_{1}}\,  d \lambda}.
\]
Since $(q_{n})$ is asymptotically optimal for $\mu$ and by  (i) we deduce
\[
\lim_{n\to \infty} \frac{M_{g}^{n}(A)}{M_{g}(A)} =
\frac
{C(r) \left( \int g^{\beta_{1}} \, d \lambda \right)^{\beta_{2}}  
\frac{\int_{A} g^{\beta_{1}}\,  d \lambda  }{\int g^{\beta_{1}}\,  d \lambda}  } 
{C(r) \left( \int g^{\beta_{1}} d \lambda \right)^{\frac{r}{1-\alpha }} 
\int_{A} g^{\beta_{1}} \, d \lambda   } = 1
\]
which proves (ii). \qed

\section{Asymptotic mismatch}

\label{sec_mismatch}

In this section we investigate the performance of a sequence of
quantizers $(q_n)$ that  is asymptotically optimal for
the source distribution $\mu$ having density $g$, but  is applied to
the  source distribution $\nu$ having density $f$. 

\begin{theorem}
\label{mresul}
Let $r>1$ and $\alpha \in (0, 1)$. 
Suppose $\mu=g\lambda$, $\nu =f\lambda$, where $g$ and $f$ are weakly
unimodal densities such that $f/g$ is bounded. Assume  $\int
|x|^{r+\delta}\,  d \mu(x) < \infty$ for some $\delta > 
0$.
If $(q_{n})_{n \in \mathbb{N}}$ is an asymptotically optimal sequence of
quantizers for $\mu$, then 
\begin{equation}
\label{relth01}
\lim_{n\to \infty} e^{(1-\alpha )( H_{\nu }^{\alpha }(q_{n}) - H_{\mu }^{\alpha }(q_{n}) )}=
\frac{\int (f/g)^{\alpha } g^{\beta_{1}}\,  d \lambda }{ \int
  g^{\beta_{1}}\,  d \lambda } 
\end{equation}
and
\begin{eqnarray}
\lefteqn{\lim_{n\to \infty}  e^{r H_{\nu }^{\alpha }(q_{n})} D_{\nu}(q_n)} \nonumber \quad \quad \\* 
&& = 
C(r) \left( \int f^{\alpha } ( g^{\frac{1}{\beta_{2}}} )^{1- \alpha
}\,  d \lambda \right)^{\frac{r}{1- \alpha }} 
\int f ( g^{\frac{1}{\beta_{2}}} )^{- r }\,  d \lambda.
\label{relth02}
\end{eqnarray}  
\end{theorem}

\begin{remark}
\label{rem_mismatch}
(a) \  The mismatch formula   (\ref{relth02}) is best interpreted through the
companding quantization approach. In  \cite[Remark 4.12]{KrLi11} it
was shown  that for a source with density $g$,  companding
quantizers having  point density  $h$ induce
high-rate asymptotics performance proportional to
\begin{equation}
\label{functio}
\left( \int g^{\alpha } h^{1-\alpha }\,  d \lambda
\right)^{\frac{r}{1-\alpha }} \int g h^{-r} \, d \lambda.
\end{equation}
Asymptotically optimal companding is obtained by setting
$h = g_{\alpha,r} = g^{1/\beta_{2}} ( \int g^{1/\beta_{2}}\, d \lambda
)^{-1}$, which is the unique minimizer of (\ref{functio}). If  the sequence of
companding quantizers with this choice of $h$ is now applied 
to the mismatched distribution $\nu = f \lambda$, then the same
asymptotic performance as in (\ref{relth02}) is obtained. Thus the
main significance of  (\ref{relth02}) is that 
it holds for an \emph{arbitrary} 
asymptotically optimal sequence $(q_n)$.
The analogy with  companding quantization  
suggests that although $(q_{n})$ can have infinitely many codecells,
one  can interpret $g_{\alpha,r}$ as
the point density related to every asymptotically optimal  sequence of
quantizers  for $\mu = g \lambda$.

\smallskip

\noindent (b)\ Using the notation introduced in Sections~\ref{seq_asymptcond} and  \ref{subsec_mismatch}, we can rewrite the
mismatch formula (\ref{relth02}) in the equivalent forms
\begin{eqnarray*}
\lim_{n\to \infty}  e^{r H_{\nu }^{\alpha }(q_{n})} D_{\nu}(q_n) &=& 
C(r)  e^{-r
   \mathcal{D}_{\alpha}(f\|g_{\alpha,r})} \int
 \frac{f}{\bigl(g_{\alpha,r}\bigr)^r} \, d\lambda \\
&=&  C(r) \|f\|_{\frac{1}{1+r}}  e^{r(
   \mathcal{D}_{1+r}(f_{0,r}\|g_{\alpha,r}) -
   \mathcal{D}_{\alpha}(f\|g_{\alpha,r}))} \\
&=& Q_{0,r}(\nu) e^{r(
   \mathcal{D}_{1+r}(f_{0,r}\|g_{\alpha,r}) -
   \mathcal{D}_{\alpha}(f\|g_{\alpha,r}))}.
\end{eqnarray*}
Formula (\ref{eq_alphaloss}) for the loss due to mismatch follows
from either of the last two expressions. The loss is always greater
than one unless $\mu=\nu$ since, according to the preceding comment,
$h=f_{\alpha,r}$ is the unique minimizer of  $ e^{r(
   \mathcal{D}_{1+r}(f_{0,r}\|h) -
   \mathcal{D}_{\alpha}(f\|h))}$ over all densities $h$. 

\smallskip

\noindent(c) The condition for the boundedness of  $f/g$ is the same as
  in the variable-rate mismatch result of \cite{GrLi03}. The
  fixed-rate result of Bucklew \cite{Buc84} requires essentially the
  same condition since   the only known example when
  the uniform   integrability condition given there  is satisfied
  requires that  $f/g$ be  bounded. 
  
\smallskip

\noindent (d) The conditions of Theorem~\ref{mresul} are satisfied
when the support of $\mu$ and $\nu$ is the same compact interval $I$
and the corresponding densities $g$ and $f$ are continuous and bounded
away from zero on $I$.  But the theorem may also apply to
distributions with unbounded support. For example, if $g$ and $f$ are
Gaussian or Laplacian densities with mean zero and variance
$\sigma_\mu^2$ and $\sigma_\nu^2$, respectively, then the conditions
are met if $\sigma_\mu^2\ge \sigma_\nu^2$.  Unfortunately, the
boundedness condition is not satisfied when
$\sigma_\mu^2<\sigma_\nu^2$ or when $g$ is Gaussian and $f$ is
Laplacian.  Na \cite{Na11} obtained a mismatch result for two
zero-mean Laplacian sources with arbitrarily mismatched variances by
considering quantile quantizers, a special class of fixed-rate
asymptotically optimal quantizers closely related to companding
quantizers.

%% From (\ref{convmism}) one recognizes that the right hand side of
%% (\ref{relth01}) and (\ref{relth02}) is finite if $f/g$ and
%% $(f/g)^{\alpha }$ are $\hat{\mu }-$integrable. It remains open if
%% Theorem \ref{mresul} still holds if we replace the boundedness of
%% $f/g$ by the $\hat{\mu }-$integrability of $f/g$ and $(f/g)^{\alpha
%% }$.

\end{remark}

\begin{proof}
Let $I = \supp (\nu )$.
We will proceed in several steps. 

\smallskip

\noindent 1. First we prove relation (\ref{relth01}) under  the stated
assumptions on $g$ and $f/g$, but additionally  assuming that  $I$ 
is a compact interval and 
\begin{equation}
\label{eqfmin}
 \min \{ f(x) : x \in I \} >0.  
\end{equation}
Let $m \geq 2$ and let $\{I_{k,m}:\, k=1,\ldots,m\}$ be a collection
  of disjoint intervals of equal length $\lambda(I)/m$ such that
  $\bigcup_{k=1}^m I_{k,m}=I$. Let $l_{k,m}=\inf I_{k,m}$ and
  $r_{k,m}=\sup I_{k,m}$ denote, respectively, the left and right
  endpoints of $I_{k,m}$. Define
\[
S_{m, n} = \bigcup_{k=1}^{m} \{ q_{n}(l_{k,m}) , q_{n}(r_{k,m})  \} \subset q_{n}(\mathbb{R})
\]
and
\[
i(f,I) = \min \{ f(x) : x \in I \}, \qquad s(f,I) = \max \{ f(x) : x \in I \}  .
\]
Note that $\card(S_{m,n}) \leq m+1$ and  $0<i(f,I)\le s(f,I)<\infty$. 
Since  $f/g \leq M$ for some $M<\infty$, we have
\begin{equation}
\label{sfig}
\nu ( A ) \leq 
M \mu ( A )
\end{equation}
for any measurable $A \subset \mathbb{R}$. 
Thus by (\ref{entrconvzero}) in Lemma~\ref{lemmaux} we get
\begin{eqnarray*}
\limsup_{n \to \infty } 
\frac{  \sum_{a \in S_{m,n}} \nu ( q_{n}^{-1}(a)  )^{\alpha } }
{ \sum_{a \in q_{n}(\mathbb{R})} \mu ( q_{n}^{-1}(a) )^{\alpha }  }  
& \leq & M^{\alpha }
\limsup_{n \to \infty } 
\frac{ \sum_{a \in S_{m,n}} \mu ( q_{n}^{-1}(a)  )^{\alpha } }
{ \sum_{a \in q_{n}(\mathbb{R})} \mu ( q_{n}^{-1}(a) )^{\alpha }  }   
= 0.
\end{eqnarray*}
Noting that for any $a\in q_n(\R)\setminus S_{m,n}$ we either have $q_n^{-1}(a) \subset
I_{k,m}$ for some $k\in\{1,\ldots,m\}$ or $\nu(q_n^{-1}(a))=0$, the above implies 
\begin{eqnarray}
\lefteqn{ \limsup_{n \to \infty } e^{(1-\alpha )( H_{\nu }^{\alpha
    }(q_{n}) - H_{\mu }^{\alpha }(q_{n}) )}}  \nonumber \\
&=&
\limsup_{n \to \infty }
\frac
{\sum_{a \in q_{n}(\mathbb{R})} \nu ( q_{n}^{-1}(a) )^{\alpha }  }
{\sum_{a \in q_{n}(\mathbb{R})} \mu ( q_{n}^{-1}(a) )^{\alpha }  } \nonumber \\
&=&
\sum_{k=1}^{m}
\limsup_{n \to \infty }  
\frac 
{\sum_{a \in q_{n}(\mathbb{R})} \nu ( q_{n}^{-1}(a) \cap I_{k,m} )^{\alpha }  }
{\sum_{a \in q_{n}(\mathbb{R})} \mu ( q_{n}^{-1}(a) )^{\alpha }  } \nonumber \\
&=&
\sum_{k=1}^{m}
\limsup_{n \to \infty }  \frac 
{\sum_{a \in q_{n}(\mathbb{R})} 
\left( \frac{\nu ( q_{n}^{-1}(a) \cap I_{k,m} )}{\mu ( q_{n}^{-1}(a) \cap I_{k,m} )} \right) ^{\alpha } 
\mu ( q_{n}^{-1}(a) \cap I_{k,m} )^{\alpha }  }
{\sum_{a \in q_{n}(\mathbb{R})} \mu ( q_{n}^{-1}(a) )^{\alpha }  }.
\label{entrlimrep}
\end{eqnarray}
Now we observe that $f/g\le M$ and (\ref{eqfmin}) imply  for all $n\ge
1$, $k\in \{1,\ldots,m\}$,  and  $a \in q_{n}(\mathbb{R})$,
\begin{equation}
0 < \frac{i(f, I_{k,m})}{s(g, I_{k,m})} 
\leq 
\frac{\nu ( q_{n}^{-1}(a) \cap I_{k,m} ) }{ \mu ( q_{n}^{-1}(a) \cap I_{k,m} ) }
\leq
\frac{s(f, I_{k,m})}{i(g, I_{k,m})} < \infty .
\label{mismeas}
\end{equation}
Combining (\ref{entrlimrep}) and (\ref{mismeas}) we deduce from
(\ref{propentrloc}) in   Theorem~\ref{theoremdens} that
\begin{eqnarray}
\limsup_{n \to \infty }  e^{(1-\alpha )( H_{\nu }^{\alpha }(q_{n}) - H_{\mu }^{\alpha }(q_{n}) )}
& \leq &
\sum_{k=1}^{m}
\left( \frac{s(f, I_{k,m})}{i(g, I_{k,m})} \right) ^{\alpha } 
\limsup_{n \to \infty } \frac 
{\sum_{a \in q_{n}(\mathbb{R})} 
\mu ( q_{n}^{-1}(a) \cap I_{k,m} )^{\alpha }  }
{\sum_{a \in q_{n}(\mathbb{R})} \mu ( q_{n}^{-1}(a) )^{\alpha }  } \nonumber \\
&=&
\sum_{k=1}^{m}
\left( \frac{s(f, I_{k,m})}{i(g, I_{k,m})} \right) ^{\alpha } 
\frac{\int_{I_{k,m}} g^{\beta_{1}} \, d \lambda }{ \int g^{\beta_{1}} \,
  d \lambda  } 
=
\int \overline{h}_{m}\,  d \hat{\mu}
\label{ubourhs}
\end{eqnarray}
where $\hat{\mu}$ is defined in (\ref{defmuhat}) and  we have defined
\[
\overline{h}_{m} = \sum_{k=1}^{m} 1_{I_{k,m}} \left(  \frac{s(f, I_{k,m})}{i(g, I_{k,m})}  \right)^{\alpha }.
\]
Here $1_{A}$ denotes  the characteristic function of $A\subset \R$
defined by $1_A(x)=1$ if $x\in A$ and $1_A(x)=0$ if $x\notin
A$. Similarly we obtain 
\begin{eqnarray}
\liminf_{n \to \infty } e^{(1-\alpha )( H_{\nu }^{\alpha }(q_{n}) - H_{\mu }^{\alpha }(q_{n}) )}
& \geq &
\int \underline{h}_{m}\,  d \hat{\mu}
\label{lbourhs}
\end{eqnarray}
with
\[
\underline{h}_{m} =
\sum_{k=1}^{m} 1_{I_{k,m}} \left(  \frac{i(f, I_{k,m})}{s(g, I_{k,m})}  \right)^{\alpha }.
\]
Obviously, $\underline{h}_{m} \leq \overline{h}_{m}$, and since 
$f$, $g$, and $f/g$  are continuous on $I$ and the common length of
the intervals $I_{k,m}$ converges to zero as $m\to \infty$, 
\[
\lim_{m\to \infty} \underline{h}_{m}(x)= \lim_{m\to \infty}
\overline{h}_{m}(x) = \bigl(f(x)/g(x)\bigr)^{\alpha }\quad \text{for all $x\in I$}.
\]
Since the $\overline{h}_m$ are uniformly bounded, from Fatou's lemma
and by dominated convergence,  we get
\begin{eqnarray}
\int (f/g)^{\alpha }\,  d \hat{\mu} &=&
\int \liminf_{m \to \infty} \underline{h}_{m}\,  d \hat{\mu} 
\leq 
\liminf_{m \to \infty} \int  \underline{h}_{m}\,  d \hat{\mu} \nonumber \\
& \leq &
\limsup_{m \to \infty} \int  \overline{h}_{m}\,  d \hat{\mu} = \int
(f/g)^{\alpha }\,  d \hat{\mu}. 
\label{limilims}
\end{eqnarray}
Combining (\ref{limilims}) with (\ref{ubourhs}) and (\ref{lbourhs}) we obtain
\[
\lim_{n\to \infty} e^{(1-\alpha )( H_{\nu }^{\alpha }(q_{n}) - H_{\mu }^{\alpha }(q_{n})
  )} = \int (f/g)^{\alpha }\,  d \hat{\mu}.
\]
By the definition  of $\hat{\mu}$ in (\ref{defmuhat}) this  yields
(\ref{relth01}). 

\smallskip 

\noindent 2. We now prove relation (\ref{relth01}) under the stated
assumptions.  Since $f$ is weakly unimodal, the set $I_{\delta} = \{x:\, f(x)\ge
\delta\}\subset I$ is a compact interval for all $\delta>0$ small
enough. Since $\bigcup_{\delta>0} I_{\delta}= I$,
  we have  $\nu ( I
  \setminus I_{\delta} ) \to 0$ as $\delta \to 0$, and we also have 
   $\hat{\mu}( I \setminus I_{\delta} ) \to 0$ as $\delta \to 0$  because
  $\hat{\mu}(\cdot|I)$ is absolutely continuous with respect to $\nu$.
 Consequently,
\begin{eqnarray}
\lefteqn{ \limsup_{n\to \infty} \biggl| \int_{I}  (f/g)^{\alpha }\,  d \hat{\mu} - \nu (I_{\delta })^{-
  \alpha } \int_{I_{\delta}}  (f/g)^{\alpha }\,  d \hat{\mu}  \biggr|}
 \nonumber \quad \quad  \\ 
&\le&     M^{\alpha } \bigl( \hat{\mu}(I \setminus I_{\delta }) + | 1 - \nu
 (I_{\delta })^{- \alpha } | \hat{\mu}(I_{\delta }) \bigr)\to 0 \label{1third}
\end{eqnarray}
as $\delta\to 0$. Set $[c_{\delta},d_{\delta}]:= I_{\delta}$. Using (\ref{sfig}) and  (\ref{entrconvzero})  in Lemma~\ref{lemmaux} we obtain 
\begin{eqnarray}
\lefteqn{
\limsup_{n \to \infty } 
\biggl| e^{(1-\alpha )( H_{\nu }^{\alpha }(q_{n})  - H_{\mu }^{\alpha }(q_{n}) )} -
e^{(1-\alpha )( H_{\nu ( \cdot | I_{\delta } ) }^{\alpha }(q_{n})  -
  H_{\mu }^{\alpha }(q_{n}) )} \biggr| }  \nonumber  \\
&=&
\limsup_{n \to \infty } 
\biggl| \frac{ \sum_{a \in q_{n}(\mathbb{R})} \nu ( q_{n}^{-1}(a) )^{\alpha } -
\nu (I_{\delta })^{- \alpha }  \sum_{a \in q_{n}(\mathbb{R})} \nu ( q_{n}^{-1}(a) \cap I_{\delta } )^{\alpha }  }
{ \sum_{a \in q_{n}(\mathbb{R})} \mu ( q_{n}^{-1}(a) )^{\alpha }   }
\biggr|  \nonumber  \\
& \leq &
M^{\alpha } \limsup_{n \to \infty } 
\frac{ \sum_{a\in q_n(\R):\, q_n^{-1}(a) \subset I \setminus I_{\delta} } 
\mu ( q_{n}^{-1}(a) )^{\alpha } }
{ \sum_{a \in q_{n}(\mathbb{R})} \mu ( q_{n}^{-1}(a) )^{\alpha }   }
 \nonumber    \\
&& 
\mbox{} + 
M^{\alpha } \limsup_{n \to \infty } 
\frac{ \mu ( q_{n}^{-1}(q_{n}( c_{\delta}  )) )^{\alpha } +
\mu ( q_{n}^{-1}(q_{n}(d_{\delta}  )) )^{\alpha } }
{ \sum_{a \in q_{n}(\mathbb{R})} \mu ( q_{n}^{-1}(a) )^{\alpha }   }
  \nonumber  \\
&& \mbox{} +
\limsup_{n \to \infty } \biggl|
\frac{ \sum_{a\in q_n(\R):\,  q_{n}^{-1}(a) \subset I_{\delta} } 
\nu ( q_{n}^{-1}(a) )^{\alpha }  - 
\nu (I_{\delta })^{- \alpha }  \sum_{a \in q_{n}(\mathbb{R})} \nu ( q_{n}^{-1}(a) \cap I_{\delta } )^{\alpha }
 }
{ \sum_{a \in q_{n}(\mathbb{R})} \mu ( q_{n}^{-1}(a) )^{\alpha }   }
\biggr|  \nonumber  \\
& = &
M^{\alpha } \limsup_{n \to \infty } 
\frac{ \sum_{a \in q_{n}(\mathbb{R}) } 
\mu ( q_{n}^{-1}(a) \cap I \setminus I_{\delta })^{\alpha } }
{ \sum_{a \in q_{n}(\mathbb{R})} \mu ( q_{n}^{-1}(a) )^{\alpha }   }
 \nonumber   \\
&&\mbox{}  +
\limsup_{n \to \infty } \biggl|
\frac{ \sum_{a:\in q_{n}(\mathbb{R}) : q_{n}^{-1}(b) \subset I_{\delta} } 
\nu ( q_{n}^{-1}(a) \cap I_{\delta } )^{\alpha }  - 
\nu (I_{\delta })^{- \alpha }  \sum_{a \in q_{n}(\mathbb{R})} \nu ( q_{n}^{-1}(a) \cap I_{\delta } )^{\alpha }
 }
{ \sum_{a \in q_{n}(\mathbb{R})} \mu ( q_{n}^{-1}(a) )^{\alpha }   }
\biggr| \nonumber   \\
& \leq &
M^{\alpha } \hat{\mu}(I \setminus I_{\delta }) \nonumber  \\
&& \mbox{} +
\limsup_{n \to \infty }\biggl|
\frac{ ( 1 - \nu (I_{\delta })^{- \alpha } )  \sum_{a\in q_n(\R): \, q_{n}^{-1}(a) \subset I_{\delta}  } 
\nu ( q_{n}^{-1}(a) \cap I_{\delta } )^{\alpha }   }
{ \sum_{a \in q_{n}(\mathbb{R})} \mu ( q_{n}^{-1}(a) )^{\alpha }
}\biggr| \nonumber   \\
&& 
 \mbox{} + 
\nu (I_{\delta })^{- \alpha } M^{\alpha } \limsup_{n \to \infty } 
\frac{ \mu ( q_{n}^{-1}(q_{n}( c_{\delta}  )) )^{\alpha } +
\mu ( q_{n}^{-1}(q_{n}( d_{\delta}  )) )^{\alpha } }
{ \sum_{a \in q_{n}(\mathbb{R})} \mu ( q_{n}^{-1}(a) )^{\alpha }   }
\nonumber  \\ 
&\leq &
M^{\alpha } ( \hat{\mu}(I \setminus I_{\delta }) + 1 - \nu (I_{\delta
})^{- \alpha }  )\to 0  \label{2third}
\end{eqnarray}
as $\delta\to 0$. Noting that the density of $\nu(\cdot|I_{\delta})$ satisfies the
condition imposed on $f$ in step 1, we obtain from this step that 
\begin{equation}
\label{3third}
\lim_{n\to \infty}  e^{(1-\alpha )( H_{\nu ( \cdot | I_{\delta } ) }^{\alpha }(q_{n})  - H_{\mu }^{\alpha }(q_{n}) )} 
= \nu (I_{\delta })^{- \alpha } \int_{I_{\delta }} (f/g)^{\alpha }\,
d \hat{\mu } 
\end{equation}
Combining  (\ref{1third}),(\ref{2third}), and (\ref{3third}) we obtain
that   given any $\varepsilon>0$ we can can
choose $\delta=\delta(\varepsilon)>0$ small enough and
$N=N(\delta,\varepsilon)$ large enough such that for all $n>
N$, 
\[
\biggl| e^{(1-\alpha )( H_{\nu }^{\alpha }(q_{n})  - H_{\mu }^{\alpha }(q_{n}) )} -
\int_{I}  (f/g)^{\alpha } \, d \hat{\mu} \biggr| < \varepsilon
\]
which yields (\ref{relth01}).

\smallskip

\noindent 3. We finish the proof by proving assertion (\ref{relth02}). 
Using definition  (\ref{defmgn}) we get
\[
e^{r H_{\nu }^{\alpha }(q_{n})} \int | x - q_{n}(x) |^{r} \, d \nu (x) =
\left( e^{(1-\alpha ) 
(H_{\nu }^{\alpha }(q_{n})-H_{\mu }^{\alpha }(q_{n}))}
\right)^{\frac{r}{1-\alpha }} \int (f/g) \, d M_{g}^{n}.
\]
Thus Corollary~\ref{corollary123} (ii) and (\ref{relth01}) yield
\[
\lim_{n\to \infty} e^{r H_{\nu }^{\alpha }(q_{n})} \int | x - q_{n}(x) |^{r}\,  d \nu (x) =
\left( \frac{\int (f/g)^{\alpha } 
g^{\beta_{1}}\,  d \lambda }{ \int g^{\beta_{1}}\,  d \lambda } \right)^{\frac{r}{1-\alpha }} 
\int (f/g)\,  d M_{g}.
\]
Using (\ref{mgedef}) we calculate
\begin{eqnarray}
\lefteqn{ \lim_{n \to \infty} e^{r H_{\nu }^{\alpha }(q_{n})} \int | x -
  q_{n}(x) |^{r} \, d \nu (x) } \nonumber \qquad  \\ 
&=&
\left( \frac{\int (f/g)^{\alpha } 
g^{\beta_{1}}\,  d \lambda }{ \int g^{\beta_{1}}\,  d \lambda } \right)^{\frac{r}{1-\alpha }} 
C(r) \left( \int_{\mathbb{R}} 
g^{\beta_{1}} \, d \lambda \right)^{\frac{r}{1 - \alpha }} 
\left( \int (f/g)g^{\beta_{1}}\,  d \lambda \right) \nonumber \\
&=& 
C(r) \left( \int f^{\alpha } 
g^{\beta_{1}-\alpha } \, d \lambda  \right)^{\frac{r}{1-\alpha }} 
\left( \int f g^{\beta_{1}-1}\,  d \lambda \right).
\label{convmism}
\end{eqnarray}
Using definition (\ref{defpar1}) it is  easy to check that
$\beta_{1} - \alpha = \frac{1 - \alpha }{ \beta_{2}}$ and  $\beta_{1}
- 1 = - \frac{r}{\beta_{2}}$, and hence  (\ref{convmism}) is equivalent to  (\ref{relth02}),
completing the proof. 
\end{proof}

\section{Concluding remarks}

\label{sec_concl}

We extended point density and mismatch results in fixed and
variable-rate asymptotic quantization theory to scalar quantization
with R\'enyi entropy constraint of order $\alpha\in (0,1)$. We showed
that the R\'enyi entropy contribution of a given interval to the
overall rate for a sequence of asymptotically optimal quantizers is
determined by the so-called entropy density of the sequence, an
analog of the traditional quantizer point density function. A dual of
this result quantifies the distortion contribution of a given region
to the overall distortion. We also proved a mismatch formula for a
sequence of asymptotically optimal R\'enyi entropy constrained scalar
quantizers. One can recover the known results for the traditional rate
definitions by formally setting $\alpha=0$ or $\alpha=1$ in our
density and mismatch results.

A natural question is whether the density and mismatch results of this
paper can be generalized to higher dimensional (vector)
quantization. To make progress in this direction, one first needs to
generalize Theorem~\ref{thmkrlmain} to higher dimensions
(cf.\ \cite[Section VIII]{KrLi11}) to obtain an analog of Zador's
fixed and variable-rate vector quantization results for R\'enyi
entropy constraint. Assuming one can prove such a result, the main
difficulty in generalizing our proofs seems to be controlling the
entropy contribution at the boundary of hypercubes (higher-dimensional
intervals). 

Another interesting question is whether the coincidence of distortion
and entropy densities described by (\ref{intcd}) in
Remark~\ref{rem_coinc} is particular to quantization with R\'enyi
entropy or is a deeper phenomenon.  In particular, one can ask whether
replacing R\'enyi's entropy with some more general information measure
(c.f.\ \cite{Csi08}) would preserve the existence of and the
special relationship between entropy and distortion densities. Answers
to  these questions  would provide a more complete understanding of
some of the  finer aspects of quantization theory.

As mentioned before, an analog of the fixed-rate point density result
of Bucklew \cite{Buc84} (see \eqref{eq_pointdensbucklew}) cannot hold
for arbitrary sequences of asymptotically optimal entropy-constrained
quantizers. However, point densities play an important role in our
intuitive understanding of the structure of optimal quantizers, and
may provide (heuristic) guidance in constructing (nearly) optimal
quantizers. Thus it would be interesting to find a framework within
which rigorous point density result can be proved for R\'enyi entropy
constrained quantization (and for traditional entropy-constrained
quantization). For the scalar case, companding quantization provides
such a framework, but for higher dimensions, the restriction to
companding usually precludes asymptotic optimality \cite{Buc81}.

\section{Appendix}

\noindent\emph{Proof of Lemma~\ref{lemmaux}.} \    We  first show
(\ref{exqkoeff}).   
The asymptotic optimality of $(q_n)$ for
$\mu$ implies that $D_{\mu}(q_n)\to 0$ as $n\to \infty$. Since $\mu$
has a density, this  yields, via Lemma~\ref{lemmaxprob} below,
the intuitively obvious fact that 
\[
\lim_{n\to \infty} \max \{ \mu ( q_{n}^{-1}(a)): a \in
  q_{n}(\mathbb{R})  \} =0. 
\]
This also
gives for  $i\in \{1,2\}$,
\begin{equation}
\label{convzmaxp2}
\lim_{n\to \infty} \max \{ \mu_i( q_{n}^{-1}(a)): a \in
  q_{n}(\mathbb{R})  \} =0. 
\end{equation}
Let $p=(p_1,p_2,\ldots)$ be a probability vector and $p_{\text{max}}=
\max\{p_i: \, i\in \mathbb{N} \}$. Since  $\alpha\in (0,1)$, we can
lower bound $ \hat{H}^{\alpha}(p)$ as 
\begin{eqnarray*}
    \hat{H}^{\alpha}(p)= \frac{1}{1-\alpha} \log \left(
    \sum\limits_{i=1}^{\infty} 
p_{i}^{\alpha} \right) &=&  \frac{1}{1-\alpha} \log
    \left(p_{\text{max}}^{\alpha} 
\sum\limits_{i=1}^{\infty} 
\biggl(\frac{p_{i}}{p_{\text{max}}}\biggr)^{\alpha} \right)\\
&\ge & \frac{1}{1-\alpha} \log
    \left(p_{\text{max}}^{\alpha} 
\sum\limits_{i=1}^{\infty} 
\frac{p_{i}}{p_{\text{max}}} \right)\\*
&= & - \log
  p_{\text{max}}.
\end{eqnarray*}
Combing this bound with (\ref{convzmaxp2}) yields
(\ref{exqkoeff}).

Next we prove  (\ref{entrconvzero})  by contradiction. 
 If  the first limit in (\ref{entrconvzero}) does not hold, then there
 is a $T>0$  
and a subsequence of $(q_{n})$, which we also denote by $(q_{n})$, such that
\begin{equation}
\label{convtot}
\lim_{n\to \infty} \frac{\mu ( q_{n}^{-1}(q_{n}(p)) )^{\alpha }}{ \sum_{a \in q_{n}(\mathbb{R})} \mu ( q_{n}^{-1}(a) )^{\alpha }  } 
= T.
\end{equation}
We have 
\begin{eqnarray}
H_{\mu }^{\alpha }(q_{n}) 
&=& \frac{1}{1 - \alpha } \log \left( \frac{\sum_{ a \in q_{n}(\mathbb{R})
      }   \mu ( q_{n}^{-1}(a) )^{\alpha }}{\mu ( q_{n}^{-1}(
  q_{n}(p) ) )^{\alpha }}   \right) \nonumber \\*
& & \mbox{}  +  \frac{\alpha }{1 - \alpha } \log \left( \mu ( q_{n}^{-1}( q_{n}(p) ) ) \right)
  \nonumber \\
& \le & \frac{1}{1 - \alpha } \log \left( \frac{\sum_{ a \in q_{n}(\mathbb{R})
     }   \mu ( q_{n}^{-1}(a) )^{\alpha }}{\mu ( q_{n}^{-1}(
  q_{n}(p) ) )^{\alpha }}   \right)  \label{entrsplit}
\end{eqnarray}
where the inequality holds since $\alpha \in (0,1)$. 
Because $(q_{n})$ is asymptotically optimal, we know that $H_{\mu
}^{\alpha }(q_{n}) \to \infty$ as $n\to \infty$.
But the right hand side of (\ref{entrsplit}) converges to a  finite limit
by assumption (\ref{convtot}), a contradiction.

Also,  (\ref{exqkoeff}) and an argument
 identical to the proof of the first limit in (\ref{entrconvzero})
 imply that for all 
 $p\in \R$ and $i\in\{1,2\}$,
\[
\lim_{n\to \infty} \frac{\mu_i ( q_{n}^{-1}(q_{n}(p)) )^{\alpha }}{ \sum_{a \in q_{n}(\mathbb{R})} \mu_i ( q_{n}^{-1}(a) )^{\alpha }  } 
= 0 
\]
which completes the proof.   \qed

\begin{lemma}
\label{lemmaxprob}
Assume  $\mu$ is a  probability measure on $\R^d$, let  $r>0$, and let
$\|\cdot\|$ be any norm on $\R^d$.
Suppose  $(q_n)$ is a sequence of $d$-dimensional vector
quantizers (mappings $q_n: \R^d \to \R^d$ with  $q_n(\R)$ at
most countable) such that 
\[
  \lim_{n\to \infty} \int_{\R^d} \|x-q_n(x)\|^r \, \mu(dx) =0.
\]
Then 
\begin{equation}
\label{convzmaxp1}
\lim_{n\to \infty} \max \{ \mu ( q_{n}^{-1}(a)): a \in
  q_{n}(\mathbb{R}^d)  \} =0
\end{equation}
if and only if $\mu$ is \emph{nonatomic}, i.e., $\mu(\{x\})=0$ for all
$x\in \R^d$. 
\end{lemma}

\begin{proof}  If $\mu(\{x\})>0$ for some $x$, then
  $\mu(q_n^{-1}(q_n(x)))\ge \mu(\{x\})$ shows that (\ref{convzmaxp1})
  cannot hold. Now assume that $\mu$ is nonatomic. We proceed
  indirectly to prove (\ref{convzmaxp1}). Since $\int  \|x-q_n(x)\|^r \,
  \mu(dx) \ge \int_{q^{-1}(a)} \|x-q_n(x)\|^r \, \mu(dx)$ for all $n$
  and $a\in q_n(\R^d)$, if (\ref{convzmaxp1}) does not hold, then
  (considering subsequences if necessary) there exist an
  $\varepsilon>0$, points  $a_n\in \R^d$, and measurable sets $A_n\subset \R^d$, such
  that 
\begin{equation}
\label{eqdandmubound}
 \lim_{n\to \infty}  \int_{A_n} \|x-a_n\|^r \, \mu(dx)
 =0, \qquad  \mu(A_n)\ge \varepsilon \text{\ for all $n$}.
  \end{equation}
Let $B(z,\delta)=\{x\in \R^d: \|x-z\|< \delta\}$ denote  the
open ball of radius $\delta> 0$ centered at $z\in \R^d$. We have
for all $\delta>0$, 
\[
\int_{A_n} \|x-a_n\|^r \, \mu(dx) \ge \delta^r \mu\bigl(A_n\setminus
B(a_n,\delta)\bigr) 
\]
which, combined with (\ref{eqdandmubound}), implies $\lim_n \mu\bigl(A_n\setminus
B(a_n,\delta)\bigr)  =0$. Thus for all $\delta>0$, 
\[
\liminf_{n\to \infty} \mu\bigl(B(a_n,\delta)\bigr) \ge \varepsilon.
\]
This immediately implies that $\{a_n: n\in \mathbb{N}\}$ is a bounded
set, since   $\limsup_{n} \|a_n\| =\infty$ would yield
$\liminf_n \mu(B(a_n,\delta)) =0$  because, as a probability measure
on $\R^d$,  $\mu$ is tight. Thus we can choose a subsequence of $(a_n)$, which we also
denote by $(a_n)$, such that $ a_n \to a \in \R^d$ as $n\to \infty$. For this
subsequence, $B(a_n,\delta) \subset B(a,2\delta)$ for all $n$ large
enough, implying, for all $\delta>0$, 
\[
\mu\bigl(  B(a,2\delta)\bigr) \ge \liminf_{n\to \infty} \mu\bigl(
B(a_n,\delta)\bigr) \ge \varepsilon.
\]
Since  $\mu(\{a\}) = \lim_{\delta \to 0} \mu\bigl(
B(a,2\delta)\bigr)$, we obtain $\mu(\{a\})\ge
\varepsilon$, which contradicts our assumption that $\mu$ is nonatomic. 
\end{proof}

\begin{lemma}
\label{lowbougen}
Let $A \geq 0$, $B \geq 0$,  $\gamma > 0$, and define $F:(0,1)\to \R$ by
\[
 F(z) = \frac{A}{z^{\gamma}} + \frac{B}{(1-z)^{\gamma}}.
\]
Then
\[
\inf \{ F(z) : z \in (0,1) \} = \left( B^{\frac{1}{1+\gamma }} + A^{\frac{1}{1+\gamma }} \right)^{1 + \gamma }.
\]
If $\min ( A,B ) > 0$, then $F(z_{0}) < F(z)$ for every $z \in
(0,1)\setminus \{ z_{0} \}$, where
\[
z_{0} = \frac{A^{\frac{1}{1+\gamma }}}{ A^{\frac{1}{1+\gamma }} + B^{\frac{1}{1+\gamma }} }.
\]
\end{lemma}
\begin{proof}
The assertion is obvious for the cases $A = 0$, $B = 0$ or $A+B=0$. Thus we can assume that  
$A>0$ and $B>0$. But in this case the assertion follows from elementary calculus.
\end{proof}

A special case of the following lemma has already been used in
\cite{Buc84}. For the reader's convenience we provide a detailed
proof.

\begin{lemma}
\label{lowbouasymp}
Let $r>1$ and $\alpha \in (0, 1)$. Let $E \subset \mathbb{R}$ be measurable. Then,
\[
\left( \int g^{\beta_{1}}\,  d \lambda \right)^{\beta_{2}} =
\inf \left \{ \frac{ \left( \int_{E} g^{\beta_{1}} \, d \lambda \right)^{\beta_{2}} }{R^{\beta_{2} - 1}} + 
\frac{ \left( \int_{\mathbb{R} \setminus E} g^{\beta_{1}}\,  d \lambda \right)^{\beta_{2}} }
{ (1-R)^{\beta_{2} - 1} } : R \in (0,1)  \right \} .
\]
If $\mu(E) \in (0,1)$, then 
\[
\left( \int g^{\beta_{1}}\,  d \lambda \right)^{\beta_{2}} < 
\frac{ \left( \int_{E} g^{\beta_{1}} \, d \lambda \right)^{\beta_{2}} }{R^{\beta_{2} - 1}} + 
\frac{ \left( \int_{\mathbb{R} \setminus E} 
g^{\beta_{1}}\,  d \lambda \right)^{\beta_{2}} }{ (1-R)^{\beta_{2} - 1} }
\]
for every $R \in (0,1) \setminus \{R_{0}\}$, where  $R_{0}=\int_{E} g^{\beta_{1}}\,  d \lambda / \int_{\mathbb{R}}
g^{\beta_{1}} \, d \lambda$.
\end{lemma}
\begin{proof}
The assertion follows from Lemma \ref{lowbougen} with  
\[
\gamma = \beta_{2}-1 = r/(1-\alpha ) > 0, 
\quad A= \left(  \int_E g^{\beta_{1}} \, d \lambda \right)^{\beta_{2} }, \text{ and }
B= \left(  \int_{\mathbb{R} \setminus E} g^{\beta_{1}}\,  d \lambda \right)^{\beta_{2} }
\]
(note that  $A>0$ and $B>0$ if $\mu(E) \in (0,1)$).  
\end{proof}

\end{document}